\documentclass{article}

\usepackage[utf8]{inputenc} 
\usepackage[T1]{fontenc}    
\usepackage{booktabs}       
\usepackage{nicefrac}       
\usepackage{microtype}      
\usepackage{times}             

\usepackage[round]{natbib}
\usepackage{amssymb,amsmath,amsthm,bbm}
\usepackage[margin=1in]{geometry}
\usepackage{verbatim,float,url,dsfont}
\usepackage{graphicx,subfigure,psfrag}
\usepackage{algorithm,algorithmic}
\usepackage{mathtools,enumitem}
\usepackage[colorlinks=true,citecolor=blue,urlcolor=blue,linkcolor=blue]{hyperref}
\usepackage{multirow}

\newtheorem{theorem}{Theorem}
\newtheorem{lemma}{Lemma}
\newtheorem{corollary}{Corollary}

\theoremstyle{definition}
\newtheorem{remark}{Remark}
\newtheorem{definition}{Definition}

\def\R{\mathbb{R}}

\def\P{\mathbb{P}}

\def\cA{\mathcal{A}}

\def\cS{\mathcal{S}}

\makeatletter
\newcommand*\rel@kern[1]{\kern#1\dimexpr\macc@kerna}
\newcommand*\widebar[1]{%
  \begingroup
  \def\mathaccent##1##2{%
    \rel@kern{0.8}%
    \overline{\rel@kern{-0.8}\macc@nucleus\rel@kern{0.2}}%
    \rel@kern{-0.2}%
  }%
  \macc@depth\@ne
  \let\math@bgroup\@empty \let\math@egroup\macc@set@skewchar
  \mathsurround\z@ \frozen@everymath{\mathgroup\macc@group\relax}%
  \macc@set@skewchar\relax
  \let\mathaccentV\macc@nested@a
  \macc@nested@a\relax111{#1}%
  \endgroup
}
\makeatother
\def\score{\mathcal{S}}

\def\quant{\mathrm{Quantile}}
\def\isim{\overset{\mathrm{i.i.d.}}{\sim}}
\def\eqd{\overset{d}{=}}
\def\d{\mathsf{d}}

\title{Conformal Prediction Under Covariate Shift}

\author{Ryan J. Tibshirani \and Rina Foygel Barber \and Emmanuel J. Cand{\`e}s \and Aaditya Ramdas }  

\date{}

\begin{document}
\maketitle

\begin{abstract}
We extend conformal prediction methodology beyond the case of exchangeable
data. In particular, we show that a weighted version of conformal prediction can
be used to compute distribution-free prediction intervals for problems in which
the test and training covariate distributions differ, but the likelihood ratio
between these two distributions is known---or, in practice, can be estimated
accurately with access to a large set of unlabeled data (test covariate
points). Our weighted extension of conformal prediction also applies more
generally, to settings in which the data satisfies a certain weighted notion of
exchangeability. We discuss other potential applications of our new conformal
methodology, including latent variable and missing data problems.
\end{abstract}

\section{Introduction}

Let $(X_i,Y_i) \in \R^d \times \R$, $i=1,\ldots,n$ denote training data that is
assumed to be i.i.d.\ from an arbitrary distribution $P$.  Given a desired
coverage rate $1-\alpha \in (0,1)$, consider the problem of constructing a 
band \smash{$\widehat{C}_n : \R^d \to \{\text{subsets of $\R$}\}$}, based on the
training data such that, for a new i.i.d.\ point $(X_{n+1},Y_{n+1})$,
\begin{equation}
\label{eq:goal} 
\P\Big\{ Y_{n+1} \in \widehat{C}_n(X_{n+1}) \Big\} \geq 1-\alpha,
\end{equation} 
where this probability is taken over the $n+1$ points $(X_i,Y_i)$,
$i=1,\ldots,n+1$ (the $n$ training points and the test point). Crucially, we
will require \eqref{eq:goal} to hold with no assumptions whatsoever on the
common distribution $P$. 

{\it Conformal prediction}, a framework pioneered by Vladimir Vovk and
colleagues in the 1990s, provides a means for achieving this goal, relying only
on exchangeablility of the training and test data. The definitive reference
is the book by \citet{vovk2005algorithmic}; see also
\citet{vovk2008tutorial,vovk2009online,vovk2013transductive,
burnaev2014efficiency}, and {\small \url{http://www.alrw.net}} for an
often-updated list of conformal prediction work by Vovk and colleagues.
Moreover, we refer to \citet{lei2014distribution,lei2018distribution} for recent  
developments in the areas of nonparametric and high-dimensional regression.  In
this work, we extend conformal prediction beyond the setting of exchangeable  
data, allowing for provably valid inference even when the training and test data
are not drawn from the same distribution.  We begin by reviewing the basics of
conformal prediction, in this section.  In Section \ref{sec:cov_shift}, we
describe an extension of conformal prediction to the setting of covariate shift,
and give supporting empirical results.  In Section \ref{sec:weight_exch}, we
cover the mathematical details behind our conformal extension.  We conclude in
Section \ref{sec:discuss}, and discuss several other possible applications of
our new conformal methodology.  

\subsection{Quantile lemma}

Before explaining the basic ideas behind conformal inference
(i.e., conformal prediction, we will use these two terms interchangeably), we
introduce some notation. We denote by $\quant(\beta; F)$ the level $\beta$
quantile of a distribution $F$, i.e., for $Z \sim F$,  
$$
\quant(\beta; F) = \inf\big\{z : \P\{Z \leq z\} \geq \beta\big\}.
$$
In our use of quantiles, we will allow for distributions $F$ on the augmented
real line, $\R \cup \{\infty\}$. For values $v_1,\ldots,v_n$, we write $v_{1:n}
= \{v_1,\ldots,v_n\}$ to denote their multiset. Note that this is unordered, and
allows for multiple instances the same element; thus in the present case, if
$v_i=v_j$ for $i\neq j$, then this value appears twice in $v_{1:n}$.  To denote
quantiles of the empirical distribution of the values $v_1,\ldots,v_n$, we
abbreviate
$$
\quant(\beta; v_{1:n}) = \quant\bigg(\beta ; \, \frac{1}{n}
\sum_{i=1}^n \delta_{v_i} \bigg),
$$
where $\delta_a$ denotes a point mass at $a$ (i.e., the distribution that places
all mass at the value $a$). The next result is a simple but key component
underlying conformal prediction. 

\begin{lemma}
\label{lem:quant}
If $V_1,\ldots,V_{n+1}$ are exchangeable random variables, then for any $\beta
\in (0,1)$, we have 
$$
\P\Big\{V_{n+1} \leq \quant\big(\beta; V_{1:n}\cup\{\infty\}\big) \Big\}  
\geq \beta.
$$
Furthermore, if ties between $V_1,\ldots,V_{n+1}$ occur with probability 
zero, then the above probability is upper bound by $\beta + 1/(n+1)$.
\end{lemma}

\begin{proof}
Consider the following useful fact about quantiles of a discrete distribution 
$F$, with support points $a_1,\ldots,a_k \in \R$: denoting $q=\quant(\beta; F)$,
if we reassign the points $a_i > q$ to arbitrary values strictly larger than
$q$, yielding a new distribution \smash{$\widetilde{F}$}, then the level $\beta$ 
quantile remains unchanged, \smash{$\quant(\beta; F) = \quant(\beta; 
\widetilde{F})$}. Using this fact,    
$$
V_{n+1} > \quant\big(\beta; V_{1:n}\cup\{\infty\}\big) 
\iff V_{n+1} > \quant\big(\beta; V_{1:(n+1)}\big),
$$
or equivalently,
\begin{equation}
\label{eq:inf_equiv}
V_{n+1} \leq \quant\big(\beta; V_{1:n}\cup\{\infty\}\big) 
\iff V_{n+1} \leq \quant\big(\beta; V_{1:(n+1)}\big).
\end{equation}
Moreover, it is straightforward to check that
$$
V_{n+1} \leq \quant\big(\beta; V_{1:(n+1)}\big) \iff 
\text{$V_{n+1}$ is among the $\lceil \beta (n+1) \rceil$ smallest of 
  $V_1,\ldots,V_{n+1}$}. 
$$
By exchangeability, the latter event occurs with probability at least $\lceil
\beta (n+1) \rceil / (n+1) \geq \beta$, which proves the lower bound; when there
are almost surely no ties, it holds with probability exactly $\lceil \beta (n+1)
\rceil / (n+1) \leq \beta + 1/(n+1)$, which proves the upper bound.
\end{proof}

\subsection{Conformal prediction}

We now return to the regression setting. Denote $Z_i=(X_i,Y_i)$, $i=1,\ldots,n$,
and $Z_{1:n}=\{Z_1,\ldots,Z_n\}$. We will use the abbreviation $Z_{-i} = Z_{1:n}
\setminus \{Z_i\}$. In what follows, we describe the construction of a 
prediction band satisfying \eqref{eq:goal}, using conformal inference, due
to \citet{vovk2005algorithmic}. We first choose a score function $\score$, whose 
arguments consist of a point $(x,y)$, and a multiset $Z$.\footnote{We
  emphasize that by defining $Z$ to be a multiset, we are treating its  
  points as unordered. Hence, to be perfectly explicit, the score function
  $\score$ cannot accept the points in $Z$ in any particular order, and it must
  take them in as unordered.  The same is true of the base algorithm $\cA$ used
  to define the fitted regression function \smash{$\widehat\mu$}, in the choice
  of absolute residual score function \eqref{eq:resid}.}
Informally, a low value of $\score((x,y), Z)$ indicates that the point $(x,y)$
``conforms'' to $Z$, whereas a high value indicates that $(x,y)$ is atypical
relative to the points in $Z$. For example, we might choose to define $\score$
by   
\begin{equation}
\label{eq:resid}
\score\big( (x,y), Z \big) = |y - \widehat\mu(x)|, 
\end{equation}
where \smash{$\widehat\mu : \R^d \to \R$} is a regression function that was
fitted by running an algorithm $\cA$ on $(x,y)$ and $Z$.

Next, at each $x \in \R^d$, we define the conformal prediction
interval\footnote{For convenience, throughout, we will refer to
  \smash{$\widehat{C}_n(x)$} as an ``interval'', even though this may actually
  be a union of multiple nonoverlapping intervals. Similarly, for simplicity, we
  will refer to \smash{$\widehat{C}_n$} as a ``band''.}
\smash{$\widehat{C}_n(x)$} by repeating the following procedure for each 
$y \in \R$. We calculate the {\it nonconformity scores}
\begin{equation}
\label{eq:scores}
V_i^{(x,y)} = \score\big( Z_i, Z_{-i} \cup \{(x,y)\} \big), \; i=1,\ldots,n, 
\quad \text{and} \quad
V_{n+1}^{(x,y)} = \score\big( (x,y), Z_{1:n} \big),
\end{equation}
and include $y$ in our prediction interval \smash{$\widehat{C}_n(x)$} if     
$$
V_{n+1}^{(x,y)} \leq \quant\big( 1-\alpha; V_{1:n}^{(x,y)} \cup \{\infty\} \big),
$$
where \smash{$V_{1:n}^{(x,y)}=\{V_1^{(x,y)},\ldots,V_n^{(x,y)}\}$}.
Importantly, the symmetry in the construction of the nonconformity scores
\eqref{eq:scores} guarantees exact coverage in finite samples.  The next theorem
summarizes this coverage result.  The lower bound is a standard result in 
conformal inference, due to \citet{vovk2005algorithmic}; the upper bound, as far
as we know, was first pointed out by \citet{lei2018distribution}. 

\begin{theorem}[\citealt{vovk2005algorithmic,lei2018distribution}]
\label{thm:conf}
Assume that $(X_i,Y_i) \in \R^d \times \R$, $i=1,\ldots,n+1$ are exchangeable. 
For any score function $\score$, and any $\alpha \in (0,1)$, define the
conformal band (based on the first $n$ samples) at $x \in \R^d$ by
\begin{equation}
\label{eq:conf}
\widehat{C}_n(x) = \Big\{ y \in \R: V_{n+1}^{(x,y)} \leq 
\quant\big(1-\alpha; V_{1:n}^{(x,y)} \cup \{\infty\} \big) \Big\}, 
\end{equation}
where \smash{$V_i^{(x,y)}$}, $i=1,\ldots,n+1$ are as defined in
\eqref{eq:scores}. Then \smash{$\widehat{C}_n$} satisfies  
$$
\P\Big\{ Y_{n+1} \in \widehat{C}_n(X_{n+1}) \Big\} \geq 1-\alpha.
$$
Furthermore, if ties between
\smash{$V_1^{(X_{n+1},Y_{n+1})},\ldots, V_{n+1}^{(X_{n+1},Y_{n+1})}$} occur with
probability zero, then this probability is upper bounded by $1-\alpha+1/(n+1)$. 
\end{theorem}

\begin{proof}
To lighten notation, abbreviate \smash{$V_i=V_i^{(X_{n+1},Y_{n+1})}$}, 
  $i=1,\ldots,n+1$. Observe
$$
Y_{n+1} \in \widehat{C}_n(X_{n+1}) \iff 
V_{n+1} \leq \quant\big(1-\alpha; V_{1:n} \cup \{\infty\} \big).
$$ 
By the symmetric construction of the nonconformity scores in \eqref{eq:scores},    
$$
(Z_1, \ldots, Z_{n+1}) \eqd (Z_{\sigma(1)}, \ldots, Z_{\sigma(n+1)})   
\iff (V_1, \ldots, V_{n+1}) \eqd (V_{\sigma(1)}, \ldots, V_{\sigma(n+1)}),  
$$
for any permutation $\sigma$ of the numbers $1,\ldots,n+1$.  Therefore, as
$Z_1,\ldots,Z_{n+1}$ are exchangeable, so are $V_1,\ldots,V_{n+1}$, and applying
Lemma \ref{lem:quant} gives the result. 
\end{proof}

\begin{remark}
Theorem \ref{thm:conf} is stated for exchangeable samples $(X_i,Y_i)$,
$i=1,\ldots,n+1$, which is (significantly) weaker than assuming i.i.d.\ samples.
As we will see in what follows, it is furthermore possible to relax the
exchangeability assumption, under an appropriate modification to the
conformal procedure.  
\end{remark}

\begin{remark}
If we use an appropriate random tie-breaking rule (to determine the rank of 
$V_{n+1}$ among $V_1,\ldots,V_{n+1}$), then the upper bounds in Lemma
\ref{lem:quant} and Theorem \ref{thm:conf} hold in general (without assuming
there are no ties almost surely).
\end{remark}

The result in Theorem \ref{thm:conf}, albeit very simple to prove, is quite
remarkable. It gives us a recipe for distribution-free prediction intervals,
with nearly exact coverage, starting from an arbitrary score function $\score$;
e.g., absolute residuals with respect to a fitted regression function from any
base algorithm $\cA$, as in \eqref{eq:resid}. For more discussion of conformal 
prediction, its properties, and its variants, we refer to
\citet{vovk2005algorithmic,lei2018distribution} and references therein.      

\section{Covariate shift}
\label{sec:cov_shift}

In this paper, we are concerned with settings in which the data $(X_i,Y_i)$,
$i=1,\ldots,n+1$ are no longer exchangeable.  Our primary focus will be a
setting in which we observe data according to 
\begin{equation}
\label{eq:cov_shift}
\begin{gathered}
(X_i, Y_i) \isim P = P_X \times P_{Y|X}, \; i=1,\ldots,n, \\
(X_{n+1}, Y_{n+1}) \sim \widetilde{P} = \widetilde{P}_X \times P_{Y|X}, \;
\text{independently}.  
\end{gathered}
\end{equation}
Notice that the conditional distribution of $Y|X$ is assumed to be the same for
both the training and test data.  Such a setting is often called {\it covariate
  shift} (e.g., see \citealt{shimodaira2000improving,quinonero2009dataset}; see 
also Remark \ref{rem:cov_shift} below for more discussion of this
literature).  The key realization is the following: if we know the ratio of test to
training covariate likelihoods, \smash{$\d \widetilde{P}_X/\d P_X$}, then we can
still perform a modified of version conformal inference, using a quantile
of a suitably weighted empirical distribution of nonconformity scores.  The next
subsection gives the details; following this, we describe a more computationally
efficient conformal procedure, and give an empirical demonstration. 

\subsection{Weighted conformal prediction}

In conformal prediction, we compare the value of a nonconformity
score at a test point to the empirical distribution of nonconformity
scores at the training points. In the covariate shift case, where the
covariate distributions \smash{$P_X,\widetilde{P}_X$} in our training
and test sets differ, we will now weight each nonconformity score
\smash{$V_i^{(x,y)}$} (which measures how well $Z_i=(X_i,Y_i)$ conforms to the
other points) by a probability proportional to the likelihood ratio
\smash{$w(X_i)=\d \widetilde{P}_X(X_i)/\d  P_X(X_i)$}. Therefore, we will no
longer be interested in the empirical distribution   
$$
\frac{1}{n+1} \sum_{i=1}^n \delta_{V_i^{(x,y)}} + \frac{1}{n+1} \delta_\infty, 
$$
as in Theorem \ref{thm:conf}, but rather, a weighted version
$$
\sum_{i=1}^n p^w_i(x) \delta_{V_i^{(x,y)}} + p^w_{n+1}(x) \delta_\infty,  
$$
where the weights are defined by
\begin{equation}
\label{eq:weight_prob_cs}
p^w_i(x) = \frac{w(X_i)}{\sum_{j=1}^n w(X_j) + w(x)}, \; i=1,\ldots,n, 
\quad \text{and} \quad 
p^w_{n+1}(x) = \frac{w(x)}{\sum_{j=1}^n w(X_j) + w(x)}. 
\end{equation}
Due this careful weighting, draws from the discrete distribution in the second
to last display resemble nonconformity scores computed on the test population, 
and thus, they ``look exchangeable'' with the nonconformity score at our test
point. Our main result below formalizes these claims.   

\begin{corollary}
\label{cor:weight_conf_cs}
Assume data from the model \eqref{eq:cov_shift}.  Assume that
\smash{$\widetilde{P}_X$} is absolutely continuous with respect to $P_X$, and 
denote \smash{$w=\d \widetilde{P}_X/\d P_X$}.  For any score function $\score$, 
and any $\alpha \in (0,1)$, define for $x \in \R^d$,
\begin{equation}
\label{eq:weight_conf_cs}
\widehat{C}_n(x) = \bigg\{ y \in \R: V_{n+1}^{(x,y)} \leq \quant\bigg(1-\alpha;
\, \sum_{i=1}^n p^w_i(x) \delta_{V_i^{(x,y)}} + p^w_{n+1}(x) \delta_\infty
\bigg) \bigg\}, 
\end{equation}
where \smash{$V_i^{(x,y)}$}, $i=1,\ldots,n+1$ are as in \eqref{eq:scores}, 
and $p_i^w(x)$, $i=1,\ldots,n+1$ are as in \eqref{eq:weight_prob_cs}. 
Then \smash{$\widehat{C}_n$} satisfies 
$$
\P\Big\{ Y_{n+1} \in \widehat{C}_n(X_{n+1}) \Big\} \geq 1-\alpha. 
$$
\end{corollary}

Corollary \ref{cor:weight_conf_cs} is a special case of a more general result
that we present later in Theorem \ref{thm:weight_conf}, which extends conformal
inference to a setting in which the data are what we call {\it weighted
  exchangeable}. The proof is given in Section \ref{sec:cor_proof}. 

\begin{remark}
\label{rem:prop_to}
The same result as in Corollary \ref{cor:weight_conf_cs} holds if \smash{$w
  \propto \d \widetilde{P}_X/\d P_X$}, i.e., with an unknown normalization constant,
because this normalization constant cancels out in the calculation of
probabilities in \eqref{eq:weight_prob_cs}.
\end{remark}

\begin{remark}
\label{rem:cov_shift}
Though the basic premise of covariate shift---and certainly the techniques
employed in addressing it---are related to much older ideas in statistics, the 
specific setup \eqref{eq:cov_shift} has recently generated great interest in 
machine learning: e.g., see
\citet{sugiyama2005input,sugiyama2007covariate,quinonero2009dataset,
  agarwal2011linear,wen2014robust,reddi2015doubly,chen2016robust} and
references therein).  The focus is usually on correcting estimators,
model evaluation, or model selection approaches to account for
covariate shift.  Correcting distribution-free prediction intervals,
as we examine in this work, is (as far as we know) a new
contribution. As one might expect, the likelihood ratio
\smash{$\d \widetilde{P}_X/\d P_X$}, a key component of our conformal
construction in Corollary \ref{cor:weight_conf_cs}, also plays a
critical role in much of the literature on covariate shift.
\end{remark}

\subsection{Weighted split conformal}
\label{sec:split_conf}

In general, constructing a conformal prediction band can be computationally
intensive, though this of course depends on the choice of score function.
Consider the use of absolute residuals as in \eqref{eq:resid}.  To compute the
nonconformity scores in \eqref{eq:scores}, we must first run our base algorithm
$\cA$ on the data set $Z_{1:n} \cup \{(x,y)\}$ to produce a fitted regression
function \smash{$\widehat\mu$}, and then calculate
$$
V_i^{(x,y)} = |Y_i - \widehat\mu(X_i)|, \; i=1,\ldots,n, 
\quad \text{and} \quad
V_{n+1}^{(x,y)} = |y - \widehat\mu(x)|. 
$$ 
As the formation of the conformal set in \eqref{eq:conf} (ordinary case) or
\eqref{eq:weight_conf_cs} (covariate shift case) requires us to do this for
each $x \in \R^d$ and $y \in \R$ (which requires refitting \smash{$\widehat\mu$}
each time), this can clearly become computationally burdernsome.   

A fast alternative, known as {\it split conformal prediction}
\citep{papadopoulos2002inductive,lei2015conformal}, resolves this issue by 
taking the score function $\score$ to be defined using absolute residuals with 
respect to a {\it fixed} regression function, typically, one that has been
trained on an preliminary data set. Denote by
\smash{$(X^0_1,Y^0_1),\ldots,(X^0_{n_0},Y^0_{n_0})$} this preliminary data set, 
used for fitting the regression function $\mu_0$, and consider the score
function  
$$
\score\big( (x,y), Z \big) = |y - \mu_0(x)|.
$$
Given data $(X_1,Y_1),\ldots,(X_n,Y_n)$, independent of
\smash{$(X^0_1,Y^0_1),\ldots,(X^0_{n_0},Y^0_{n_0})$}, we calculate 
$$
V_i^{(x,y)} = |Y_i - \mu_0(X_i)|, \; i=1,\ldots,n, 
\quad \text{and} \quad
V_{n+1}^{(x,y)} = |y - \mu_0(x)|.
$$
The conformal prediction interval \eqref{eq:conf}, defined at a point $x \in
\R^d$, reduces to   
\begin{equation}
\label{eq:conf_split}
\widehat{C}_n(x) = \mu_0(x) \pm \quant\Big(1-\alpha ; 
\big\{|Y_i - \mu_0(X_i)|\big\}_{i=1}^n \cup \{\infty\} \Big),
\end{equation}
and by Theorem \ref{thm:conf} it has coverage at least $1-\alpha$, 
conditional on \smash{$(X^0_1,Y^0_1),\ldots,(X^0_{n_0},Y^0_{n_0})$}.  
This coverage result holds because, when we treat $\mu_0$ as fixed (meaning,
condition on \smash{$(X^0_1,Y^0_1),\ldots,(X^0_{n_0},Y^0_{n_0})$}), the scores
\smash{$V_1^{(x,y)},\ldots,V_{n+1}^{(x,y)}$} scores are exchangeable for
$(x,y)=(X_{n+1},Y_{n+1})$, as $(X_1,Y_1),\ldots,(X_{n+1},Y_{n+1})$ are.  

As split conformal prediction can be seen as a special case of conformal
prediction, in which the regression function $\mu_0$ is treated as fixed,
Corollary \ref{cor:weight_conf_cs} also applies to the split scenario, and
guarantees that the band defined for $x \in \R^d$ by
\begin{equation}
\label{eq:weight_conf_cs_split}
\widehat{C}_n(x) = \mu_0(x) \pm  \quant\bigg(1-\alpha; 
\, \sum_{i=1}^n p^w_i(x) \delta_{|Y_i -\mu_0(X_i)|} + p^w_{n+1}(x)
\delta_\infty \bigg),
\end{equation}
where the probabilities are as in \eqref{eq:weight_prob_cs}, has coverage at
least $1-\alpha$, conditional on
\smash{$(X^0_1,Y^0_1),\ldots,(X^0_{n_0},Y^0_{n_0})$}.    

\subsection{Airfoil data example} 

We demonstrate the use of conformal prediction in the covariate shift
setting in an empirical example.  We consider the airfoil data set from the UCI
Machine Learning Repository \citep{dua2017uci}, which has $N=1503$ observations 
of a response $Y$ (scaled sound pressure level of NASA airfoils), and a
covariate $X$ with $d=5$ dimensions (log frequency, angle of attack, chord
length, free-stream velocity, and suction side log displacement thickness).
For efficiency, we use the split conformal prediction methods in
\eqref{eq:conf_split} and \eqref{eq:weight_conf_cs_split} for the unweighted and
weighted case, respectively. R code to reproduce these simulation results can 
be found at {\small \url{http://www.github.com/ryantibs/conformal/}}.

\paragraph{Creating training data, test data, and covariate shift.}

We repeated an experiment for 5000 trials, where for each trial we randomly 
partitioned the data \smash{$\{(X_i,Y_i)\}_{i=1}^N$} into three 
sets $D_{\text{pre}},D_{\text{train}},D_{\text{test}}$, and also constructed a
covariate shift test set $D_{\text{shift}}$, which have the following roles. 

\begin{itemize}
\item $D_{\text{pre}}$, containing 25\% of the data, is used to prefit
  a regression function $\mu_0$ via linear regression, to be used in
  constructing split conformal intervals. 
\item $D_{\text{train}}$, containing 25\% of the data, is our training set, 
  i.e., $(X_i,Y_i)$, $i=1,\ldots,n$, used to compute the residual quantiles in 
  the construction of the split conformal prediction intervals
  \eqref{eq:conf_split} and \eqref{eq:weight_conf_cs_split}.
\item $D_{\text{test}}$, containing 50\% of the data, is our test set (as
  these data points are exchangeable with those in $D_{\text{train}}$, there is
  no covariate shift in this test set).  
\item $D_{\text{shift}}$ is a second test set, constructed to simulate covariate
  shift, by sampling 25\% of the points from $D_{\text{test}}$ with replacement,  
  with probabilities proportional to    
  \begin{equation}
  \label{eq:exp_tilt}
  w(x) = \exp(x^T \beta), \quad \text{where} \quad \beta=(-1,0,0,0,1).  
  \end{equation}
\end{itemize}

As the original data points $D_{\text{train}} \cup D_{\text{test}}$ can be
seen as draws from the same underlying distribution, we can view $w(x)$ as
the likelihood ratio of covariate distributions between the test set
$D_{\text{shift}}$ and training set $D_{\text{train}}$.  Note that the test
covariate distribution \smash{$\widetilde{P}_X$}, which satisfies \smash{$\d
  \widetilde{P}_X \propto \exp(x^T \beta) \d P_X$} as we have defined it here,
is called an {\it exponential tilting} of the training covariate distribution 
$P_X$.  Figure \ref{fig:tilt} visualizes the effect of this exponential tilting,
in the airfoil data set, with our choice $\beta=(-1,0,0,0,1)$.  Only the 1st and 
5th dimensions of the covariate distribution are tilted; the bottom row of
Figure \ref{fig:tilt} plots the marginal densities of the 1st and 5th covariates
(estimated via kernel smoothing) before and after the tilt.  The top row plots
the response versus the 1st and 5th covariates, simply to highlight the fact
that there is heteroskedasticity, and thus we might expect the shift in the 
covariate distribution to have some effect on the validity of the ordinary 
conformal prediction intervals.

\begin{figure}[htb]
\centering
\includegraphics[width=0.75\textwidth]{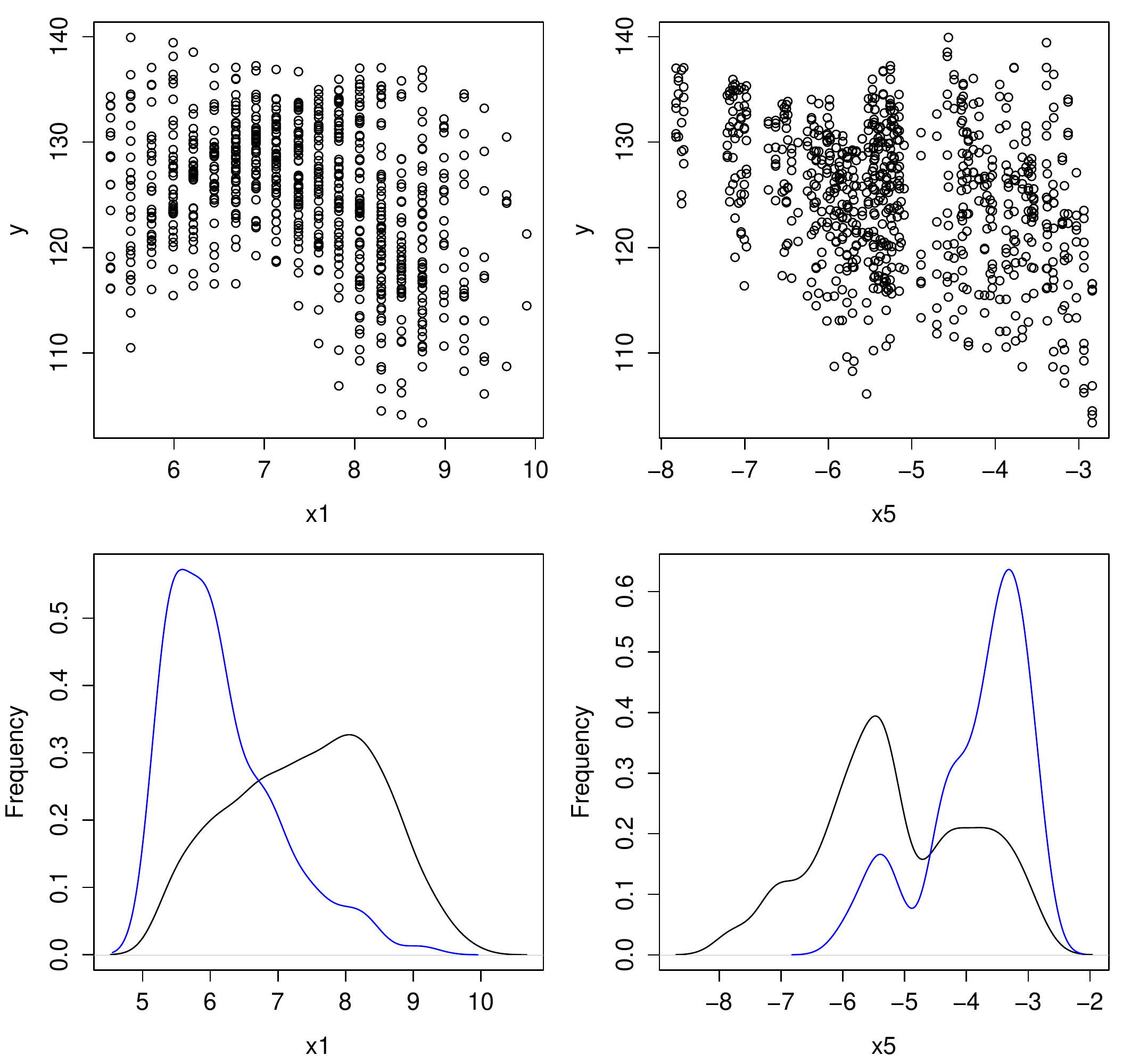}
\caption{\small The top row plots the response in (a randomly chosen half of) 
  the airfoil data set, versus the 1st and 5th covariates.  The bottom row plots
  kernel density estimates for the 1st and 5th covariates, in black. Also
  displayed are kernel density estimates for the 1st and 5th covariates after 
  exponential tilting \eqref{eq:exp_tilt}, in blue.} 
\label{fig:tilt}
\end{figure}

\paragraph{Loss of coverage of ordinary conformal prediction under covariate
  shift.} 

First, we examine the performance of ordinary split conformal prediction 
\eqref{eq:conf_split}. The nominal coverage level was set to be 90\% (meaning  
$\alpha=0.1$), here and throughout. The results are shown in the top row of 
Figure \ref{fig:cov}.  In each of the 5000 trials, we computed the empirical
coverage from the split conformal intervals over points in the test sets, and
the histograms show the distribution of these empirical coverages over the 
trials. We see that for the original test data $D_{\text{test}}$ (no
covariate shift, shown in red), split conformal works as expected, with the
average of the empirical coverages (over the 5000 trials) being 90.2\%;
but for the nonuniformly subsampled test data $D_{\text{shift}}$ (covariate
shift, in blue), split conformal considerably undercovers, with its average
coverage being 82.2\%.     

\begin{figure}[p]
\centering
\includegraphics[width=0.9\textwidth]{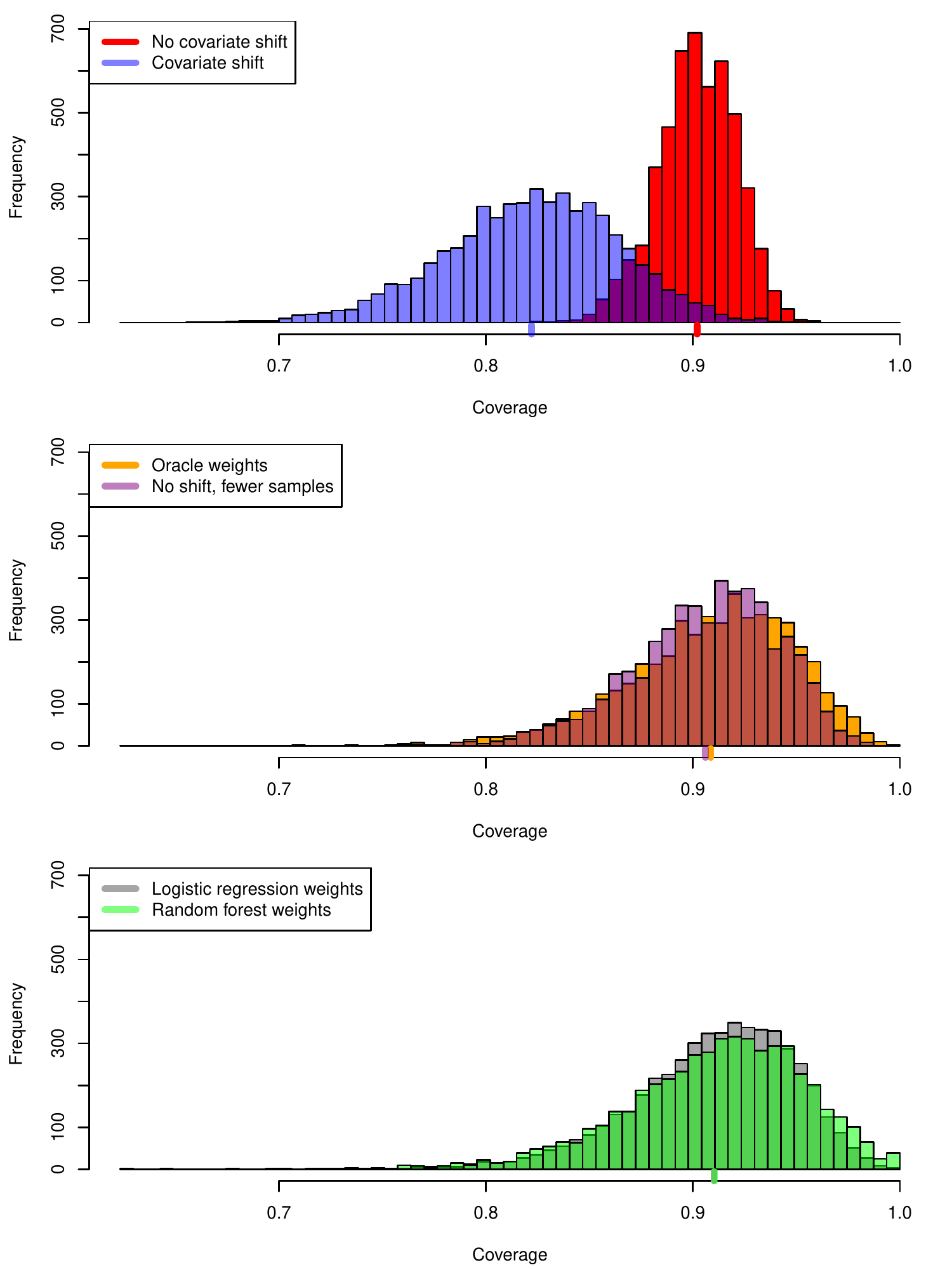}
\caption{\small Empirical coverages of conformal prediction intervals, computed    
  using 5000 different random splits of the airfoil data set.  The averages of
  empirical coverages in each histogram are marked on the x-axis.}  
\label{fig:cov}
\end{figure}

\paragraph{Coverage of weighted conformal prediction with oracle weights.} 

Next, displayed in the middle row of Figure \ref{fig:cov}, we consider weighted  
split conformal prediction \eqref{eq:weight_conf_cs_split}, to cover the
points in $D_{\text{shift}}$ (shown in orange). At the moment, we assume oracle
knowledge of the true weight function $w$ in \eqref{eq:exp_tilt} needed to
calculate the probabilities in \eqref{eq:weight_prob_cs}.  We see that this
brings the coverage back to the desired level, with the average coverage being
90.8\%. However, the histogram is more dispersed than it is when there  
is no covariate shift (compare to the top row, in red).  This is because, by
using a quantile of the weighted empirical distribution of nonconformity scores, 
we are relying on a reduced ``effective sample size''.  Given training points 
$X_1,\ldots,X_n$, and a likelihood ratio $w$ of test to training covariate
distributions, a popular heuristic formula from the covariate shift literature
for the effective sample size of $X_1,\ldots,X_n$ is
\citep{gretton2009covariate,reddi2015doubly}: 
$$
\widehat{n}= \frac{[\sum_{i=1}^n |w(X_i)|]^2}{\sum_{i=1}^n |w(X_i)|^2}
= \frac{\|w(X_{1:n})\|_1^2}{\|w(X_{1:n})\|_2^2},
$$
where we abbreviate $w(X_{1:n})=(w(X_1),\ldots,w(X_n)) \in \R^n$.
To compare weighted conformal prediction against the unweighted method at the
same effective sample size, in each trial, we ran unweighted split conformal
on the original test set $D_{\text{test}}$, but we used only
\smash{$\widehat{n}$} subsampled points from $D_{\text{train}}$ 
to compute the quantile of nonconformity scores, when constructing the 
prediction interval \eqref{eq:conf_split}.  The results (the middle
row of Figure \ref{fig:cov}, in purple) line up very closely with those from
weighted conformal, which demonstrates that apparent overdispersion in the
coverage histogram from the latter is fully explained by the reduced effective
sample size.    

\begin{figure}[p]
\centering
\includegraphics[width=0.9\textwidth]{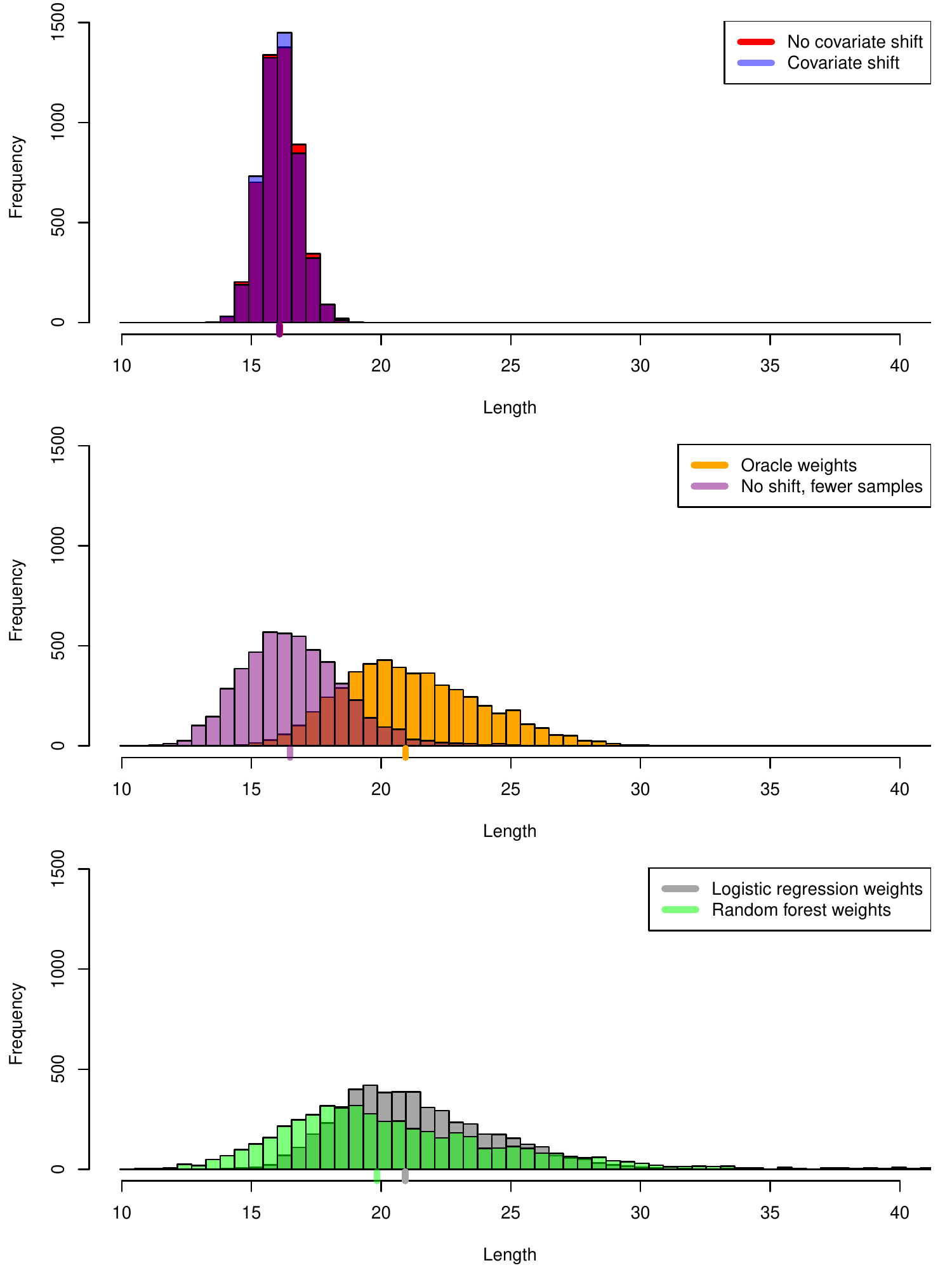}
\caption{\small Median lengths of conformal prediction intervals, computed
  using 5000 different random splits of the airfoil data set. The averages of
  median lengths in each histogram are marked on the x-axis.} 
\label{fig:len}
\end{figure}

\begin{figure}[tb]
\centering
\includegraphics[width=0.725\textwidth]{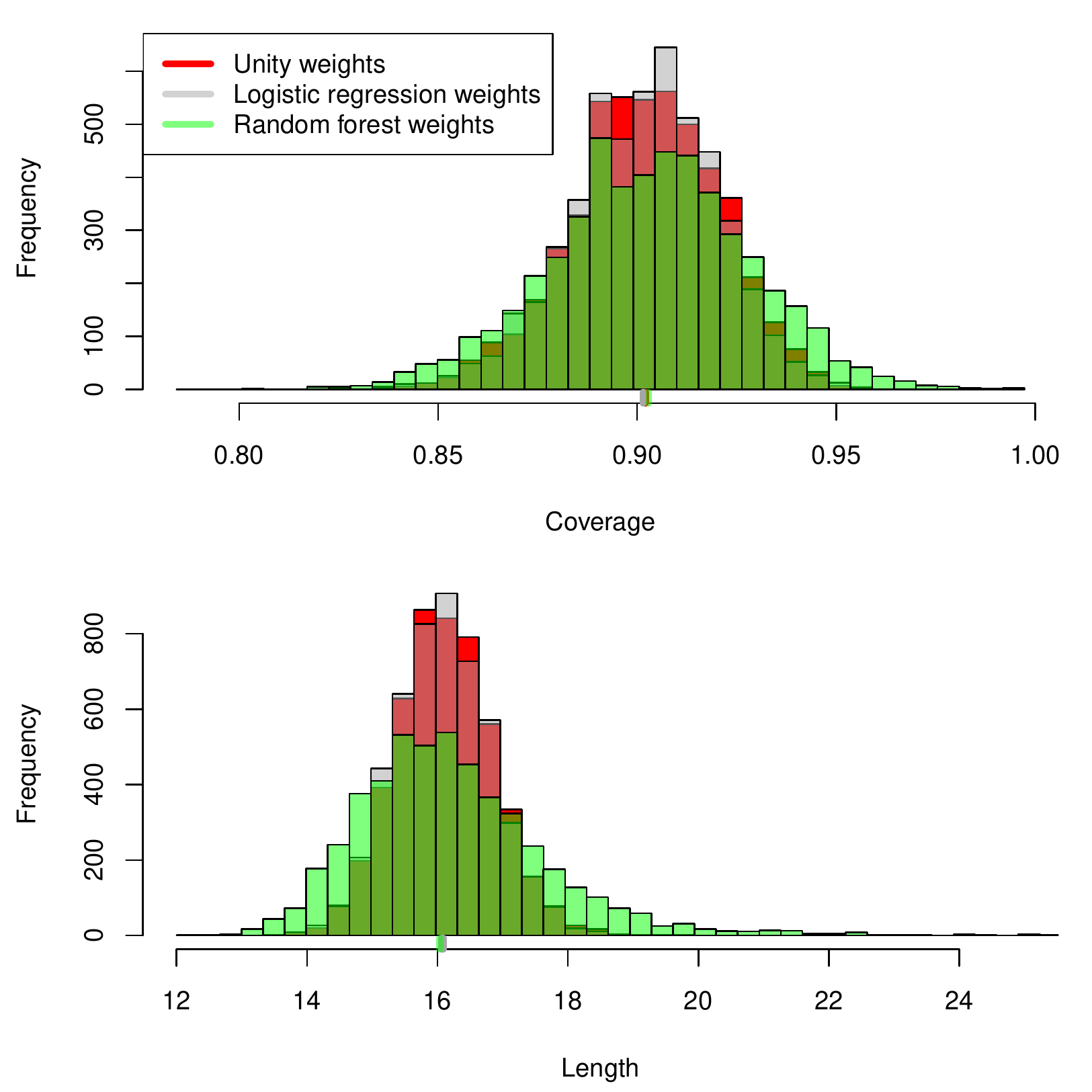}
\caption{\small Empirical coverages and median lengths from conformal  
  prediction, on the airfoil data set, with no covariate shift.}
\label{fig:ada}
\end{figure}

\paragraph{Coverage of weighted conformal with estimated weights.} 

Denote by $X_1,\ldots,X_n$ the covariate points in $D_{\text{train}}$ and 
by $X_{n+1},\ldots,X_{n+m}$ the covariate points in $D_{\text{shift}}$. 
Here we describe how to estimate \smash{$w=\d \widetilde{P}_X/\d P_X$}, the 
likelihood ratio of interest, by applying logistic regression or random 
forests (more generally, any classifier that outputs estimated probabilities
of class membership) to the feature-class pairs $(X_i,C_i)$, $i=1,\ldots,n+m$, 
where $C_i=0$ for $i=1,\ldots,n$ and $C_i=1$ for $i=n+1,\ldots,n+m$.  Noting
that 
$$
\frac{\P(C=1 | X=x)}{\P(C=0 | X=x)}
= \frac{\P(C=1)}{\P(C=0)} \frac{\d \widetilde{P}_X}{\d P_X}(x),
$$
we can take the conditional odds ratio $w(x)=\P(C=1|X=x)/\P(C=0|X=x)$ as an
equivalent representation for the oracle weight function (since we only need to
know the likelihood ratio up to a proportionality constant, recall Remark
\ref{rem:prop_to}). Therefore, if \smash{$\widehat{p}(x)$} is an estimate of
$\P(C=1|X=x)$ obtained by fitting a classifier to the data $(X_i,C_i)$,
$i=1,\ldots,n+m$, then we can use  
\begin{equation}
\label{eq:weight_est}
\widehat{w}(x) = \frac{\widehat{p}(x)}{1-\widehat{p}(x)}
\end{equation}
as our estimated weight function for the calculation of probabilities 
\eqref{eq:weight_prob_cs}, needed for conformal \eqref{eq:weight_conf_cs} or
split conformal \eqref{eq:weight_conf_cs_split}.  There is in fact a sizeable 
literature on density ratio estimation, and the method just describe falls
into a class called {\it probabilistic classification} approaches; two other
classes are based on moment matching, and minimization of $\phi$-divergences 
(e.g., Kullblack-Leibler divergence).  For a comprehensive review of these
approaches, and supporting theory, see \citet{sugiyama2012density}.   

The bottom row of Figure \ref{fig:cov} shows the results from using weighted
split conformal prediction to cover the points in $D_{\text{shift}}$, where the
weight function \smash{$\widehat{w}$} has been estimated as in
\eqref{eq:weight_est}, using logistic regression (in gray) and random
forests\footnote{In the random forests approach, we clipped the estimated 
  test class probability \smash{$\widehat{p}(x)$} to lie in between 0.01 and
  0.99, to prevent the estimated weight (likelihood ratio)
  \smash{$\widehat{w}(x)$} from being infinite.  Without clipping, the estimated 
  probability of being in the test class was sometimes exactly 1 (this occurred
  in about 2\% of the cases encountered over all 5000 repetitions), resulting in
  an infinite weight, and causing numerical issues.} 
(in green) to fit the class probability function \smash{$\widehat{p}$}. Note
that logistic regression is well-specified in this example, because it assumes 
the log odds is a linear function of $x$, which is exactly as in
\eqref{eq:exp_tilt}.  Random forests, of course, allows more flexibility in the
fitted model. Both classification approaches deliver weights that translate into  
good average coverage, being 91.0\% for each approach.  Furthermore, their 
histograms are only a little more dispersed than that for the oracle weights
(middle row, in orange).   

\paragraph{Lengths of weighted conformal intervals.} 

Figure \ref{fig:len} conveys the same setup as Figure \ref{fig:cov}, but displays
histograms of the median lengths of prediction intervals rather than empirical  
coverages (meaning, in each of the 5000 trials, we ran unweighted or weighted
split conformal prediction to cover test points, and report the median length of
the prediction intervals over the test sets).  We see no differences in the
lengths of ordinary split conformal intervals (top row) when there is or is not
covariate shift, as expected since these two settings differ only in the
distributions of their test sets, but use the same procedure and have the same
distribution of the training data. We see that the oracle-weighted split
conformal intervals are longer than the ordinary split conformal intervals
that use an equivalent effective sample size (middle row).  This is also
as expected, since in the former situation, the regression function
$\mu_0$ was fit on training data $D_{\text{train}}$ of a different distribution
than $D_{\text{shift}}$, and $\mu_0$ itself should ideally be adjusted to
account for covariate shift (plenty of methods for this are available from the
covariate shift literature, but we left it unadjusted for simplicity). Lastly,
we see that the random forests-weighted split conformal intervals are more
variable, and in some cases, much longer, than the logistic regression-weighted
split conformal intervals (bottom row, difficult to confirm visually because the
bars in the histogram lie so close to the x-axis).  

\paragraph{Weighted conformal when there is actually no covariate shift.}

Lastly, Figure \ref{fig:ada} compares the empirical coverages and median
lengths of split conformal intervals to cover points in $D_{\text{test}}$ (no
covariate shift), using the ordinary unweighted approach (in red), the logistic
regression-weighted approach (in gray), and the random forests-weighted approach
(in green).  The unweighted and logistic regression approaches are very similar.
The random forests approach yields slightly more dispersed coverages and
lengths.  This is because random forests are very flexible, and in the present
case of no covariate shift, the estimated weights from random forests in each
repetition are in general further from constant (compared to those from logistic 
regression).  Still, random forests must not be overfitting dramatically here,
since the coverages and lengths are still reasonable.

\section{Weighted exchangeability}
\label{sec:weight_exch}

In this section, we develop a general result on conformal prediction for
settings in which the data satisfy what we call {\it weighted exchangeability}.
In the first subsection, we take a look back at the key quantile result in Lemma  
\ref{lem:quant}, and present an alternative proof from a somewhat different
perspective.  Then we precisely define weighted exchangeability, extend 
Lemma \ref{lem:quant} to this new (and broader) setting, and extend 
conformal prediction as well.  The last subsection provides a proof of Corollary
\ref{cor:weight_conf_cs}, as a special case of our general conformal result.  

\subsection{Alternate proof of Lemma \ref{lem:quant}}
\label{sec:new_proof}

The general strategy we pursue here is to condition on the unlabeled multiset of 
values obtained by our random variables $V_1,\ldots,V_{n+1}$, and then inspect
the probabilities that the last random variable $V_{n+1}$ attains each one of
these values. For simplicity, we assume that there are almost surely no ties
among the scores $V_1,\ldots,V_{n+1}$, so that we can work with sets rather than
multisets (our arguments apply to the general case as well, but the notation is
more cumbersome).     

Denote by $E_v$ the event that $\{V_1,\ldots,V_{n+1}\}=\{v_1,\ldots,v_{n+1}\}$,
and consider 
\begin{equation}
\label{eq:last_prob}
\P\{V_{n+1}=v_i \,|\, E_v\}, \; i=1,\ldots,n+1.
\end{equation}
Denote by $f$ the probability density function\footnote{More generally, $f$ may
  be the Radon-Nikodym derivative with respect to an arbitrary base measure;
  this measure may be discrete, continuous, or neither; henceforth, we will use
  the term ``density'' just for simplicity.} of the 
joint distribution $V_1,\ldots,V_{n+1}$. Exchangeability implies that 
$$
f(v_1,\ldots,v_{n+1}) = f(v_{\sigma(1)},\ldots, v_{\sigma(n+1)})
$$ 
for any permutation $\sigma$ of the numbers $1,\ldots,n+1$. Thus, for each $i$,
we have   
\begin{align}
\nonumber
\P\{V_{n+1}=v_i \,|\, E_v\}
&= \frac{\sum_{\sigma : \sigma(n+1)=i} f(v_{\sigma(1)},\ldots, v_{\sigma(n+1)})}  
{\sum_\sigma f(v_{\sigma(1)},\ldots, v_{\sigma(n+1)})} \\
\nonumber
&= \frac{\sum_{\sigma: \sigma(n+1)=i} f(v_1,\ldots,v_{n+1})}  
{\sum_\sigma f(v_1,\ldots,v_{n+1})} \\
\label{eq:exch_calc}
&= \frac{n!}{(n+1)!} = \frac{1}{n+1}. 
\end{align}
This shows that the distribution of $V_{n+1}|E_v$ is uniform on the set 
$\{v_1,\ldots,v_{n+1}\}$, i.e.,
$$
V_{n+1}|E_v \sim \frac{1}{n+1}\sum_{i=1}^{n+1}\delta_{v_i},
$$
and it follows immediately that
$$
\P\bigg\{ V_{n+1} \leq \quant\bigg(\beta; \,
\frac{1}{n+1}\sum_{i=1}^{n+1}\delta_{v_i} \bigg) \, \bigg| \,
E_v\bigg\} \geq \beta.
$$
On the event $E_v$, we have $\{V_1,\ldots,V_{n+1}\} = \{v_1,\ldots,v_{n+1}\}$, 
so 
$$
\P\bigg\{ V_{n+1} \leq \quant\bigg(\beta; \, 
\frac{1}{n+1}\sum_{i=1}^{n+1}\delta_{V_i}\bigg) \, \bigg| \, 
E_v\bigg\} \geq \beta.
$$
Because this true for any $v$, we can marginalize to obtain   
$$
\P\bigg\{ V_{n+1} \leq \quant\bigg(\beta; \,
\frac{1}{n+1}\sum_{i=1}^{n+1}\delta_{V_i}\bigg) \bigg\} \geq \beta,
$$
which, as argued in \eqref{eq:inf_equiv}, is equivalent to the desired lower
bound in the lemma. (The upper bound follows similarly.)

\subsection{Weighted exchangeability}

The alternate proof in the last subsection is certainly more complicated than
our first proof of Lemma \ref{lem:quant}.  What good did it do? Recall, we
proceeded as if we had observed the unordered set of nonconformity scores
$\{V_1,\ldots,V_{n+1}\}=\{v_1,\ldots,v_{n+1}\}$ but had forgotten the labels
(as if we forgot which random variable $V_i$ was associated with which value
$v_j$).  We then reduced the construction of a prediction interval for $V_{n+1}$
to the computation of probabilities in \eqref{eq:last_prob}, that the last
random variable $V_{n+1}$ attains each one of the observed values
$v_1,\ldots,v_{n+1}$. The critical point was that this strategy isolated the
role of exchangeability: it was used precisely to compute these probabilities,
in \eqref{eq:exch_calc}.  As we will show, the same probabilities in
\eqref{eq:last_prob} can be calculated in a broader setting of interest,
beyond the exchangeable one. 

We first define a generalized notion of exchangeability.

\begin{definition}
\label{def:weight_exch}
Random variables $V_1,\ldots,V_n$ are said to be {\it weighted exchangeable}, 
with weight functions $w_1,\ldots,w_n$, if the density\footnote{As before, 
  $f$ may be the Radon-Nikodym derivative with respect to an arbitrary base 
  measure.} $f$ of their joint distribution can be factorized as 
$$
f(v_1,\ldots,v_n) = \prod_{i=1}^n w_i(v_i) \cdot g(v_1,\ldots,v_n), 
$$
where $g$ is any function that does not depend on the ordering of its inputs,
i.e., \smash{$g(v_{\sigma(1)}, \ldots, v_{\sigma(n)}) = g(v_1,\ldots,v_n)$} for
any permutation $\sigma$ of $1,\ldots,n$. 
\end{definition}

Clearly, weighted exchangeability with weight functions $w_i \equiv 1$ for
$i=1,\ldots,n$ reduces to ordinary exchangeability.  Furthermore, independent
draws (where all marginal distributions are absolutely continuous with respect
to, say, the first one) are always weighted exchangeable, with weight functions
given by the appropriate Radon-Nikodym derivatives, i.e., likelihood ratios.
This is stated next; the proof follows directly from Definition
\ref{def:weight_exch} and is omitted.

\begin{lemma}
\label{lem:weight_exch}
Let $Z_i \sim P_i$, $i=1,\ldots,n$ be independent draws, where each $P_i$ is 
absolutely continuous with respect to $P_1$, for $i \geq 2$.  Then
$Z_1,\ldots,Z_n$ are weighted exchangeable, with weight functions 
$w_1 \equiv 1$, and $w_i=\d P_i/\d P_1$, $i \geq 2$.
\end{lemma}

Lemma \ref{lem:weight_exch} highlights an important special case (which
we note, includes the covariate shift model in \eqref{eq:cov_shift}).  But it is 
worth being clear that our definition of weighted exchangeability encompasses
more than independent sampling, and allows for a nontrivial dependency structure
between the variables, just as exchangeability is broader than the i.i.d. case.     

\subsection{Weighted quantile lemma}

Now we give a weighted generalization of Lemma \ref{lem:quant}.  

\begin{lemma}
\label{lem:weight_quant}
Let $Z_i$, $i=1,\ldots,n+1$ be weighted exchangeable random variables, with 
weight functions $w_1,\ldots,w_{n+1}$.  Let $V_i=\score( Z_i, Z_{-i} )$, where
$Z_{-i}=Z_{1:(n+1)}\setminus\{Z_i\}$, for $i=1,\ldots,n+1$, and $\score$ is an 
arbitrary score function. Define
\begin{equation}
\label{eq:weight_prob}
p^w_i(z_1,\ldots,z_{n+1}) = 
\frac{\sum_{\sigma : \sigma(n+1)=i} \prod_{j=1}^{n+1} w_j(z_{\sigma(j)})}
{\sum_\sigma \prod_{j=1}^{n+1} w_j(z_{\sigma(j)})}, \; i=1,\ldots,n+1,  
\end{equation} 
where the summations are taken over permutations $\sigma$ of the numbers
$1,\ldots,n+1$.  Then for any $\beta\in(0,1)$,
$$
\P\bigg\{ V_{n+1} \leq \quant\bigg(\beta; \, \sum_{i=1}^n 
p^w_i(Z_1,\ldots,Z_{n+1}) \delta_{V_i} + p^w_{n+1}(Z_1,\ldots,Z_{n+1}) 
\delta_\infty\bigg) \bigg\} \geq \beta.
$$
\end{lemma}

\begin{proof}
We follow the same general strategy in the alternate proof of Lemma
\ref{lem:quant} in Section \ref{sec:new_proof}.  As before, we assume for
simplicity that $V_1,\ldots,V_{n+1}$ are distinct almost surely (but the result
holds in the general case as well).

Let $E_z$ denote the event that $\{Z_1,\ldots,Z_{n+1}\}=\{z_1,\ldots,z_{n+1}\}$,
and let $v_i=\score(z_i, z_{-i} )$, where $z_{-i}=z_{1:(n+1)}\setminus\{z_i\}$,
for $i=1,\ldots,n+1$.  Let $f$ be the density function of the joint sample
$Z_1,\ldots,Z_{n+1}$. For each $i$,  
$$
\P\{V_{n+1} = v_i \,|\, E_z\} = \P\{Z_{n+1} = z_i \,|\, E_z\}
= \frac{\sum_{\sigma : \sigma(n+1)=i} f(z_{\sigma(1)},\ldots, z_{\sigma(n+1)})} 
{\sum_\sigma f(z_{\sigma(1)},\ldots, z_{\sigma(n+1)})},
$$
and as $Z_1,\ldots,Z_{n+1}$ are weighted exchangeable,
\begin{align*}
\frac{\sum_{\sigma : \sigma(n+1)=i} f(z_{\sigma(1)},\ldots, z_{\sigma(n+1)})}  
{\sum_\sigma f(z_{\sigma(1)},\ldots, z_{\sigma(n+1)})} &= 
\frac{\sum_{\sigma : \sigma(n+1)=i} \prod_{j=1}^{n+1} w_j(z_{\sigma(j)}) \cdot 
  g(z_{\sigma(1)},\ldots,z_{\sigma(n+1)})}
{\sum_\sigma \prod_{j=1}^{n+1} w_j(z_{\sigma(j)}) \cdot
  g(z_{\sigma(1)},\ldots,z_{\sigma(n+1)})}  \\
&= \frac{\sum_{\sigma : \sigma(n+1)=i} \prod_{j=1}^{n+1} w_j(z_{\sigma(j)})
  \cdot g(z_1,\ldots,z_{n+1})} 
{\sum_\sigma\prod_{j=1}^{n+1} w_j(z_{\sigma(j)}) \cdot g(z_1,\ldots,z_{n+1})} \\ 
&= p^w_i(z_1,\ldots,z_{n+1}).
\end{align*}
In other words,  
$$
V_{n+1}|E_z \sim \sum_{i=1}^{n+1} p^w_i(z_1,\ldots,z_{n+1})\delta_{v_i},
$$
which implies that
$$
\P\bigg\{ V_{n+1} \leq \quant\bigg(\beta; \, \sum_{i=1}^{n+1}
p^w_i(z_1,\ldots,z_{n+1})\delta_{v_i}\bigg) \, \bigg| \, E_z\bigg\}
\geq \beta. 
$$
This is equivalent to  
$$
\P\bigg\{ V_{n+1} \leq \quant\bigg(\beta; \, \sum_{i=1}^{n+1} 
p^w_i(Z_1,\ldots,Z_{n+1}) \delta_{V_i}\bigg) \, \bigg| \, E_z\bigg\}
\geq \beta, 
$$
and after marginalizing, 
$$
\P\bigg\{ V_{n+1} \leq \quant\bigg(\beta; \, \sum_{i=1}^{n+1}
p^w_i(Z_1,\ldots,Z_{n+1}) \delta_{V_i}\bigg) \bigg\} \geq \beta. 
$$
Finally, as in \eqref{eq:inf_equiv}, this is equivalent to the claim in the
lemma. 
\end{proof}

\begin{remark}
When $V_1,\ldots,V_{n+1}$ are exchangeable, we have $w_i \equiv 1$ for
$i=1,\ldots,n$, and thus $p^w_i\equiv 1$ for $i=1,\ldots,n$ as well. In this
special case, then, the lower bound in Lemma \ref{lem:weight_quant} reduces to
the ordinary unweighted lower bound in Lemma \ref{lem:quant}.  Meanwhile, a
obtaining meaningful upper bound on the probability in question in Lemma
\ref{lem:weight_quant}, as was done in Lemma \ref{lem:quant} (under the
assumption of no ties, almost surely), does not seem possible without further
conditions on the weight functions in consideration. This is because the largest
``jump'' in the discontinuous cumulative distribution function of $V_{n+1} |
E_z$ is of size \smash{$\max_{i=1,\ldots,n+1} p^w_i(z_1,\ldots,z_{n+1})$}, which
can potentially be very large; by comparison, in the unweighted case, this jump
is always of size $1/(n+1)$.
\end{remark}

\subsection{Weighted conformal prediction}

A weighted version of conformal prediction follows immediately from Lemma
\ref{lem:weight_quant}. 

\begin{theorem}
\label{thm:weight_conf}
Assume that $Z_i=(X_i,Y_i) \in \R^d \times \R$, $i=1,\ldots,n+1$ are weighted  
exchangeable with weight functions $w_1,\ldots,w_{n+1}$. For any score
function $\score$, and any $\alpha \in (0,1)$, define the weighted conformal
band (based on the first $n$ samples) at a point $x \in \R^d$ by
\begin{equation}
\label{eq:weight_conf}
\widehat{C}_n(x) = \bigg\{ y \in \R: V_{n+1}^{(x,y)} \leq \quant\bigg(1-\alpha;
\, \sum_{i=1}^n p^w_i\big(Z_1,\ldots,Z_n,(x,y)\big) \delta_{V_i^{(x,y)}} \,+\, 
p^w_{n+1}\big(Z_1,\ldots,Z_n,(x,y)\big) \delta_\infty\bigg)\bigg\},  
\end{equation}
where \smash{$V_i^{(x,y)}$}, $i=1,\ldots,n+1$ are as in \eqref{eq:scores}, 
and $p^w_i$, $i=1,\ldots,n+1$ are as in \eqref{eq:weight_prob}. 
Then \smash{$\widehat{C}_n$} satisfies    
$$
\P\Big\{ Y_{n+1} \in \widehat{C}_n(X_{n+1}) \Big\} \geq 1-\alpha. 
$$
\end{theorem}

\begin{proof}
Abbreviate \smash{$V_i=V_i^{(X_{n+1},Y_{n+1})}$}, $i=1,\ldots,n+1$. By
construction,  
$$
Y_{n+1} \in \widehat{C}_n(X_{n+1}) \iff 
V_{n+1} \leq  \quant\bigg(1-\alpha; \; \sum_{i=1}^n p^w_i(Z_1,\ldots,Z_{n+1}) 
\delta_{V_i} + p^w_{n+1}(Z_1,\ldots,Z_{n+1}) \delta_\infty\bigg),
$$ 
and applying Lemma \ref{lem:weight_quant} gives the result.  
\end{proof}

\begin{remark}
As explained in Section \ref{sec:split_conf}, the split conformal method
is simply a special case of the conformal prediction framework, where we take
the score function to be $\cS((x,y),Z) = |y-\mu_0(x)|$, with $\mu_0$ precomputed 
on a preliminary data set \smash{$(X^0_1,Y^0_1),\ldots,(X^0_{n_0},Y^0_{n_0})$}. 
Hence, the result in Theorem \ref{thm:weight_conf} carries over to the split
conformal method as well, in which case the weighted conformal prediction
interval in \eqref{eq:weight_conf} simplifies to 
$$
\widehat{C}_n(x) = \mu_0(x)\pm  \quant\bigg(1-\alpha; \, 
\sum_{i=1}^n p^w_i\big(Z_1,\ldots,Z_n,(x,y)\big) \delta_{|Y_i - \mu_0(X_i)|}
+ p^w_{n+1}\big(Z_1,\ldots,Z_n,(x,y)\big) \delta_\infty \bigg). 
$$
By Theorem \ref{thm:weight_conf}, this has coverage at least $1-\alpha$,
conditional on \smash{$(X^0_1,Y^0_1),\ldots,(X^0_{n_0},Y^0_{n_0})$}.
\end{remark}

\subsection{Proof of Corollary \ref{cor:weight_conf_cs}}
\label{sec:cor_proof}

We return to the case of covariate shift, and show that Corollary
\ref{cor:weight_conf_cs} follows from the general weighted conformal
result. By Lemma \ref{lem:weight_exch}, we know that the independent draws $Z_i
= (X_i,Y_i)$, $i=1,\ldots,n+1$ are weighted exchangeable with $w_i\equiv 1$ for 
$i=1,\ldots,n$, and $w_{n+1}((x,y))=w(x)$. In this special case, the
probabilities in \eqref{eq:weight_prob} simplify to  
$$
p^w_i(z_1,\ldots,z_{n+1}) = \frac{\sum_{\sigma : \sigma(n+1)=i}
  w(x_i)}{\sum_\sigma w(x_{\sigma(n+1)})} =
\frac{w(x_i)}{\sum_{j=1}^{n+1}w(x_j)}, \; i=1,\ldots,n+1,
$$ 
in other words, \smash{$p^w_i(Z_1,\ldots,Z_n,(x,y)) = p^w_i(x)$}, 
$i=1,\ldots,n+1$, where the latter are as in \eqref{eq:weight_prob_cs}. Applying
Theorem \ref{thm:weight_conf} gives the result.    

\section{Discussion}
\label{sec:discuss}

We described an extension of conformal prediction to handle weighted
exchangeable data.  This covers exchangeable data, and independent (but not
identically distributed) data, as special cases.  In general, the
new weighted methodology requires computing quantiles of a weighted discrete 
distribution of nonconformity scores, which is combinatorially hard.  But the
computations simplify dramatically for a case of significant practical
interest, in which the test covariate distribution \smash{$\widetilde{P}_X$}
differs from the training covariate distribution $P_X$ by a known likelihood
ratio \smash{$\d \widetilde{P}_X/\d P_X$} (and the conditional distribution  
$P_{Y|X}$ remains unchanged).  In this case, known as covariate
shift, the new weighted conformal prediction methodology is just as easy, 
computationally, as ordinary conformal prediction.  When the likelihood ratio
\smash{$\d \widetilde{P}_X/\d P_X$} is not known, it can be estimated given access
to unlabeled data (test covariate points), which we showed empirically, on a
low-dimensional example, can still yield correct coverage.  

Beyond the setting of covariate shift that we have focused on (as the
main application in this paper), our weighted conformal methodology can be
applied to several other closely related settings, where ordinary conformal
prediction will not directly yield correct coverage. We discuss three such
settings below.    

\paragraph{Graphical models with covariate shift.} 

Assume that the training data $(Z,X,Y) \sim P$ has the Markovian
structure $Z \to X \to Y$. As an example, to make matters concrete, suppose $Z$
is a low-dimensional covariate (such as ancenstry information), $X$ is a
high-dimensional set of features for a person (such as genetic measurements),
and $Y$ is a real-valued outcome of interest (such as life expectancy). Suppose
that on the test data \smash{$(Z,X,Y) \sim \widetilde{P}$}, the distribution of
$Z$ has changed, causing a change in the distribution of $X$, and thus causing a
change in the distribution of the unobserved $Y$ (however the distribution of 
$X|Z$ is unchanged). One plausible solution to this problem would be to just
ignore $Z$ in both training and test sets, and run weighted conformal prediction
on only $(X,Y)$, treating this like a usual covariate shift problem. But, as $X$
is high-dimensional, this would require estimating a ratio of two
high-dimensional densities, which would be difficult. As $Z$ is
low-dimensional, we can instead estimate the weights by estimating the
likelihood ratio of $Z$ between test and training sets, which follows because
for the joint covariate $(Z,X)$,
$$
\frac{\widetilde{P}_{Z,X}(z,x)}{P_{Z,X}(z,x)} = 
\frac{\widetilde{P}_Z(z) P_{X|Z=z}(x)}{P_Z(z)P_{X|Z=z}(x)} =
\frac{\widetilde{P}_Z(z)}{P_Z(z)}. 
$$
This may be a more tractable quantity to estimate for the purpose of weighted
conformal inference. These ideas may be generalized to more complex graphical 
settings, which we leave to future work.

\paragraph{Missing covariates with known summaries.} 

As another concrete example, suppose that hospital A has collected a private
training data set $(Z,X,Y) \sim P^A$ where $Z \in \{0,1\}$ is a sensitive
patient covariate, $X \in \R^d$ represents other covariates, and $Y \in \R$ is a
response that is expensive to measure.  Hospital B also has its own data set,
but in order to save money and not measure the responses for their patients, it
asks hospital A for help to produce prediction intervals for these
responses. Instead of sharing the collected data $(Z,X) \sim P^B$ for each
patient with hospital A, due to privacy concerns, hospital B only provides
hospital A with the $X$ covariate for each patient, along with a summary
statistic for $Z$, representing the fraction of $Z$ values that equal one (more
accurately, the probability of drawing a patient with $Z=1$ from their
underlying patient population). Assume that \smash{$P^A_{X|Z}=P^B_{X|Z}$} (e.g.,
if $Z$ is the sex of the patient, then this assumes there is one joint
distribution on $X$ for males and one for females, which does not depend on
the hospital). The likelihood ratio of covariate distributions thus again
reduces to calculating the likelihood ratio of $Z$ between $P^B$ and $P^A$,
which we can easily do, and use weighted conformal prediction.

\paragraph{Towards local conditional coverage?} 

We finish by describing how the weighted conformal prediction methodology can be
used to construct prediction bands that satisfy a certain approximate
(locally-smoothed) notion of conditional coverage.  Given i.i.d.\ points
$(X_i,Y_i)$, $i=1,\ldots,n+1$, consider instead of our original goal
\eqref{eq:goal},  
\begin{equation}
\label{eq:cond}
\P\Big\{ Y_{n+1} \in \widehat{C}_n(x_0) \, \Big| \, X_{n+1}=x_0 \Big\} 
\geq 1-\alpha.  
\end{equation}
This is (exact) conditional coverage at $x_0 \in \R^d$.  As it turns out,
maintaining that \eqref{eq:cond} hold for almost all\footnote{This is meant to
  be interpreted with respect to $P_X$, i.e., for all $x_0 \in \R^d$ except on a
  set whose probability under $P_X$ is zero.} $x_0 \in \R^d$ and all
distributions $P$ is far too strong:
\citet{vovk2012conditional,lei2014distribution} prove that any method with such 
a property must yield an interval \smash{$\widehat{C}_n(x_0)$} with infinite
expected length at any non-atom point\footnote{This is a point $x_0$ in the
  support of $P_X$ such that \smash{$P_X\{B_r(x_0)\} \to 0$} as $r \to 0$, where 
  \smash{$B_r(x_0)$} is the ball of radius $r$ centered at $x_0$.} $x_0$, for
any underlying distribution $P$.  

Thus we must relax \eqref{eq:cond} and seek some notion of approximate
conditional coverage, if we hope to achieve it with a nontrivial prediction
band. Some relaxations were recently considered in \citet{barber2019limits},
most of which were also impossible to achieve in a nontrivial way. A different,
but natural relaxation of \eqref{eq:cond} is 
\begin{equation}
\label{eq:local}
\frac{\int \P\big( Y_{n+1} \in \widehat{C}_n(x_0) \, | \, X_{n+1}=x \big)
K\big(\frac{x-x_0}{h}\big)\, \d P_X(x)}
{\int K\big(\frac{x-x_0}{h}\big) \, \d P_X(x)} \geq 1-\alpha, 
\end{equation}
where $K$ is kernel function and $h>0$ is bandwidth parameter.  Here we are
asking for a prediction band whose average conditional coverage, in some
locally-weighted sense around $x$, is at least $1-\alpha$.  We can equivalently
write \eqref{eq:local} as      
\begin{equation}
\label{eq:local2}
\P\Big\{ Y_{n+1} \in \widehat{C}_n(x_0) \, \Big| \, X_{n+1} = x_0 + h\omega
\Big\} \geq 1-\alpha, 
\end{equation}
where the probability is taken over the $n+1$ data points and an independent
draw $\omega$ from a distribution whose density is proportional to $K$.  (For
example, when we take $K(x) = \exp(-\|x\|_2^2/2)$, the Gaussian kernel, we have
$\omega \sim N(0,I)$.) As we can see from \eqref{eq:local} (or
\eqref{eq:local2}), this kind of locally-weighted guarantee should be close to a
guarantee on conditional coverage, when the bandwidth $h$ is small.

In order to achieve \eqref{eq:local} in a distribution-free manner, we can
invoke the weighted conformal inference methodology.  In particular, note that
we can once more rewrite \eqref{eq:local2} as
\begin{equation}
\label{eq:local3}
\P_{x_0} \Big\{ Y_{n+1} \in \widehat{C}_n(\widetilde{X}_{n+1}) \Big\} \geq
1-\alpha,   
\end{equation}
where the probability is taken over training points $(X_i,Y_i)$, $i=1,\ldots,n$,
that are i.i.d.\ from $P = P_X \times P_{Y|X}$ and an independent test point
\smash{$(\widetilde{X}_{n+1},Y_{n+1})$}, from \smash{$\widetilde{P} =
\widetilde{P}_X \times P_{Y|X}$}, where \smash{$\d \widetilde{P}_X/\d P_X 
\propto K((\cdot - x_0)/h)$}.  The subscript $x_0$ on the probability operator
in \eqref{eq:local3} emphasizes the dependence of the test covariate
distribution on $x_0$.  Note that this precisely fits into the covariate shift
setting \eqref{eq:cov_shift}. To be explicit, for any score function $\score$,
and any $\alpha \in (0,1)$, given a center point $x_0 \in \R^d$ of interest,
define
$$
\widehat{C}_n(x) = 
\Bigg\{ y \in \R: V_{n+1}^{(x,y)} \leq \quant\Bigg(1-\alpha; \, 
\frac{\sum_{i=1}^n K\big(\frac{X_i-x_0}{h}\big) \delta_{V_i^{(x,y)}} + 
K\big(\frac{x-x_0}{h}\big) \delta_\infty}
{\sum_{i=1}^n K\big(\frac{X_i-x_0}{h}\big) + 
K\big(\frac{x-x_0}{h}\big)} \Bigg) \Bigg\}, 
$$
where \smash{$V_i^{(x,y)}$}, $i=1,\ldots,n+1$, are as in \eqref{eq:scores}. 
Then by Corollary \ref{cor:weight_conf_cs},   
\begin{equation}
\label{eq:local4}
\P_{x_0} \Big\{ Y_{n+1} \in \widehat{C}_n(X_{n+1}; x_0) \Big\} \geq 1-\alpha.   
\end{equation}
This is ``almost'' of the desired form \eqref{eq:local3} (equivalently
\eqref{eq:local}, or \eqref{eq:local2}), except for one important caveat. 
The prediction band \smash{$\widehat{C}_n(\,\cdot\, ; x_0)$} in
\eqref{eq:local4} was constructed based on the center point $x_0$ (as suggested
by our notation) to have the property in \eqref{eq:local4}.  If we were to ask
for local conditional coverage at a new point $x_0$, then the entire band
\smash{$\widehat{C}_n(\cdot; x_0)$} must change (must be recomputed) in order to
accommodate the new guarantee.  

In other words, we have not provided a recipe for the construction of a single
band \smash{$\widehat{C}_n$} based on the training data $(X_i,Y_i)$, $i=1,\ldots,n$
that has the property \eqref{eq:local3} at all $x_0 \in \R^d$ (or almost
all $x_0 \in \R^d$). We have only described a way to satisfy such a property if
the center point $x_0$ were specified in advance. 
Combined with our somewhat pessimistic results in \citet{barber2019limits}, 
many important practical and philosophical problems in assumption-lean
conditional predictive inference remain open.  

\subsection*{Acknowledgements}

The authors thank the American Institute of Mathematics for supporting and
hosting our collaboration. R.F.B.\ was partially supported by the National
Science Foundation under grant DMS-1654076 and by an Alfred P.\ Sloan
fellowship.  E.J.C.\ was partially supported by the Office of Naval Research
under grant N00014-16-1-2712, by the National Science Foundation under grant
DMS-1712800, and by a generous gift from TwoSigma.  R.J.T.\ was partially
supported by the National Science Foundation under grant DMS-1554123.

\bibliographystyle{plainnat}
\bibliography{refs}

\end{document}